\documentclass[aps,prl,twocolumn]{revtex4-2}
\usepackage{blindtext}

\usepackage{amsmath,amssymb}
\usepackage{graphicx}
\usepackage{physics}
\usepackage{tcolorbox}
\usepackage{booktabs}
\usepackage{amsthm}
\usepackage{algorithm}
\usepackage{enumitem}
\usepackage{algpseudocode}

\usepackage{tabularx}
\usepackage{booktabs}
\usepackage{newfloat}
\usepackage{colortbl}
\usetikzlibrary{positioning,shapes.misc,arrows.meta}
\usepackage{tikz}
\usetikzlibrary{positioning,fit,arrows.meta,shapes,calc}

\DeclareFloatingEnvironment[
fileext=loa,
listname=List of Algorithms,
name= Protocol]{protocol}

\usepackage{pgfplots}
\pgfplotsset{compat=1.18}
\usepackage{multirow}
\usepackage{array}

\newtheorem{theorem}{Theorem}
\definecolor{darkblue}{rgb}{0.0,0.0,0.5}
\usepackage[colorlinks=true,
            linkcolor=blue,
            citecolor=darkblue,
            urlcolor=darkblue]{hyperref}

\definecolor{darkgreen}{RGB}{0,100,0}

\begin{document}

\title{Device Independent Quantum Secret Sharing Using\\ Multiparty Pseudo-telepathy Game}
\author{Santanu Majhi}
\email{santanum\_r@isical.ac.in}
\affiliation{ Indian Statistical Institute, Kolkata, India}

\author{Goutam Paul}
\email{goutam.paul@isical.ac.in}
\affiliation{ Indian Statistical Institute, Kolkata, India}

\begin{abstract}
Device-independent quantum secret sharing (DI-QSS) provides security against untrusted quantum devices. While device-independent quantum key distribution (DI-QKD) using Mermin-Peres magic square game [Zhen \emph{et al.}, Phys. Rev. Lett, 2023] has been proposed, we present the first DI-QSS protocol based on a pseudo-telepathy parity game without requiring dedicated rounds and specific basis configuration, unlike CHSH-based DI-QSS schemes [Zhang \emph{et al.}, Phys. Rev. A, 2024]. Our scheme simultaneously certifies device-independence and key generation using a single test, achieving optimal performance for a seven-qubit GHZ state configuration. Security against collective attacks is analyzed, and a positive key rate is obtained under white noise and photon loss. Moreover, we show a bitwise advantage over the previous protocol for producing the same raw key length.
\end{abstract}
 \maketitle      
Quantum secret sharing (QSS) involves a dealer to share a secret message $\mathfrak{m}$ among $\mathcal{P}$ participants such that the message can only be reconstructed from all, or a qualified subset of the shares. Individual shares reveal no information about $\mathfrak{m}$, thereby ensuring information theoretic security. The essential components \textit{secret sharing} and \textit{secret reconstruction} together guarantee secure distribution and recovery of the secret~\cite{PhysRevLett.83.648}. Hillery \textit{et al.}~\cite{PhysRevA.59.1829}, first proposed a quantum secret sharing scheme using Greenberger-Horne-Zeilinger (GHZ) states to distribute a secret among three parties: Alice, Bob, and Charlie. The correlation structure among three users' bases (see details in Appendix~\ref{Analysis of Hillery et al.}) plays a crucial role in advancing the next-generation quantum secret sharing (QSS) schemes.

Imperfect or untrusted quantum devices motivated device-independent quantum key distribution (DI-QKD) protocols~\cite{PhysRevLett.113.140501,acin2006bell, acin2006efficient, acin2007device, pironio2009device, masanes2011secure, masanes2014full}, that guarantee security without trusting the devices.
Extending these ideas, it becomes crucial to investigate whether QSS protocols can be made device-independent.  
The first DI-QSS protocol was proposed by Roy and Mukhopadhyay~\cite{PhysRevA.100.012319} for arbitrary even dimensions, establishing correctness and completeness with respect to measurement devices. Subsequently, Moreno \textit{et al.}~\cite{PhysRevA.101.052339} introduced a DI-QSS scheme based on stronger forms of Bell nonlocality, enhancing robustness against adversaries. More recently, Zhang \textit{et al.} presented two DI-QSS schemes: one employing noise preprocessing and postselection~\cite{PhysRevA.110.042403} for practical applications, and another introducing a refined method for advanced random key generation~\cite{PhysRevA.111.012603}.

Notably, those works~\cite{PhysRevA.111.012603, PhysRevA.110.042403} introduced a larger combination of measurement bases and fixed a set of combinations for DI verification and another set of bases for key generation, while the other combinations are discarded. This naturally raises the question of whether the DI-QSS scheme can be realized with fewer resources, without requiring dedicated rounds and specific basis combinations. In this Letter, we propose a DI-QSS scheme based on multiparty pseudo-telepathy parity game to achieve device-independence. Additionally,  the single correlation test is sufficient for key generation also. This significantly lowers round inefficiency without adding extra assumptions beyond isolated secure labs for the parties, reliable classical post-processing, unrestricted measurement options, verified public communication, and an adversary constrained only by quantum physics.
Moreover, we perform security analysis of our protocol against collective attacks and also analyze its robustness under white noise.
For both cases, we obtained positive key rates. Finally, we derive that our scheme has a bitwise advantage over the previous protocol~\cite{PhysRevA.110.042403} while yielding the same raw key length.


\emph{Multiparty Pseudo-telepathy game.$-$}
A well known game~\cite{10.1007/978-3-540-45078-8_1} introduced by Mermin, and later studied in various contexts~\cite{brassard2004minimum, brassard2004recasting, greenberger1990bell, mermin1990s,peres2002quantum, renner2004quantum, mermin1990quantum, cleve2004consequences}, is the multiparty Pseudo-telepathy game involving $n$ players. Its general formulation is as follows.

Each player $j$ receives an input, a bit string $x_j \in \{ 0,1\}^l $, which is also interpreted as an integer in binary, with the promise that $\sum_j x_j$ is divided by $2^l.$ The player must output a single bit $y_j$ and the winning condition is 
$$\sum_{j=1}^{n} y_j = \frac{1}{2^l} \hspace{1.3mm}{\sum_{j=1}^{n} x_j} \hspace{4mm} (mod \hspace{2mm}2).$$

Currently, we restrict the game to $l=1$~\cite{brassard2004minimum, brassard2004recasting} \hspace{2mm}(i.e. the input $x_j$ is also a single bit).\\
\noindent For an $n$-party Pseudo-telepathy game, the players are denoted $A_{1},A_{2},\dots,A_{n}$. The inputs $x_{j}\in\{0,1\}$ must satisfy the constraint  
$
\sum_{j=1}^{n} x_{j} = 0 \pmod{2} $
i.e., the input string $x_{1}x_{2}\dots x_{n}$ contains an even number of ones.  
The winning condition requires that the outputs $y_{j}\in\{0,1\}$ obey  
\begin{align}
\sum_{j=1}^{n} y_{j} = \tfrac{1}{2}\sum_{j=1}^{n} x_{j} \pmod{2}.
\label{Winning condition}
\end{align}  
The best known classical strategy achieves a success probability of  
\(\tfrac{1}{2}+2^{-\lfloor n/2 \rfloor}\),  
whereas players employing a quantum strategy (see Appendix~\ref{Quantum strategy for the Pseudo-telepathy game} for Quantum strategy) achieve perfect success.

\emph{$3$-party {DI-QSS} using Pseudo-telepathy Game.$-$}
 Consider a $3$-party protocol where Alice, Bob, and Charlie are three parties.
Now the scheme should be designed to incorporate two phases simultaneously: $(i)$ DI-checking phase, 
   $(ii)$ Share and Reconstruction phase.
The DI-checking phase ensures device independence, while the share and reconstruction phase distributes the secret among the parties and enables its recovery. The key challenge lies in designing the protocol so that these two phases operate seamlessly in parallel.

\begin{figure}[h]
  \hspace*{-0.7cm}
\begin{tikzpicture}[
  font=\sffamily,
  >=Stealth, thick,
  node distance=0.9cm and 1.5cm,
  every node/.style={align=center}
]

\node[draw, rounded corners=4pt, fill=red!10, minimum width=1.6cm, minimum height=0.6cm] (ghz) {$\ket{\mathrm{GHZ}}$};

\node[draw, rounded corners=3pt, fill=yellow!10, minimum width=1.5cm, minimum height=0.6cm,
      below left=1.4cm and 2cm of ghz] (alice) {Alice};
      
\node[draw, rounded corners=3pt, fill=blue!10, minimum width=1.5cm, minimum height=0.6cm, below=1.4cm of ghz] (bob) {Bob};
\node[draw, rounded corners=3pt, fill=purple!10, minimum width=1.5cm, minimum height=0.6cm, below right=1.4cm and 2cm of ghz] (charlie) {Charlie};

\draw[->, thick, >=Latex] (ghz.west) -| ($(alice.north)+(0,0.2)$);
\draw[->, thick, >=Latex] (ghz.south)      -- ++(0,-0.4) -| ($(bob.north)+(0,0.2)$);
\draw[->, thick, >=Latex] (ghz.east) -| ($(charlie.north)+(0,0.2)$);

\node[below=0.35cm of alice] (x1) {$x_1 \stackrel{\$}{\leftarrow}\{0,1\}$};
\node[below=0.35cm of bob] (x2) {$x_2 \stackrel{\$}{\leftarrow}\{0,1\}$};
\node[below=0.35cm of charlie] (x3) {$x_3 \stackrel{\$}{\leftarrow}\{0,1\}$};

\node[draw, fill=yellow!20, minimum width=1.5cm, minimum height=0.6cm, below=0.35cm of x1] (bba) {$BB^A_1$};
\node[draw, fill=blue!20, minimum width=1.5cm, minimum height=0.6cm, below=0.35cm of x2] (bbb) {$BB^B_2$};
\node[draw, fill=purple!20, minimum width=1.5cm, minimum height=0.6cm, below=0.35cm of x3] (bbc) {$BB^C_3$};

\node[draw, dotted, inner sep=1pt, fit=(x1)(bba)] (boxA) {};
\node[draw, dotted, inner sep=1pt, fit=(x2)(bbb)] (boxB) {};
\node[draw, dotted, inner sep=1pt, fit=(x3)(bbc)] (boxC) {};

\draw[->] (x1) -- (bba);
\draw[->] (x2) -- (bbb);
\draw[->] (x3) -- (bbc);

\draw[->, dotted, >=Latex] ($(alice.south)+(0,-0.05)$) -- ($(bba.north)+(0,0.05)$);
\draw[->, dotted, >=Latex ] ($(bob.south)+(0,-0.05)$) -- ($(bbb.north)+(0,0.05)$);
\draw[->, dotted, >=Latex] ($(charlie.south)+(0,-0.05)$) -- ($(bbc.north)+(0,0.05)$);

\node[below=0.4cm of bba] (y1) {$y_1$};
\node[below=0.4cm of bbb] (y2) {$y_2$};
\node[below=0.4cm of bbc] (y3) {$y_3$};

\draw[->] (bba) -- (y1);
\draw[->] (bbb) -- (y2);
\draw[->] (bbc) -- (y3);

\node[font=\footnotesize\itshape, color=red, align=center]
  (noAB) at ($(bba)!0.5!(bbb)$) {No\\interaction};

\node[font=\footnotesize\itshape, color=red, align=center]
  (noBC) at ($(bbb)!0.5!(bbc)$) {No\\interaction};

\node[below=3cm of bob, align=center] {\normalsize{Output depends on received qubit and classical inputs.}};

\end{tikzpicture}
\caption{ A schematic diagram of a three-party pseudo-telepathy game. The double-headed arrows indicate that the boxes $BB_{1}^{A}$, $BB_{2}^{B}$, $BB_{3}^{C}$ receive both quantum and classical inputs, while their outputs are purely classical. Once the classical inputs are chosen, no interaction is allowed among the parties. }
\label{fig:Pseudo-telepathy_game}
\end{figure}

The case $n = 3$ is being considered to make a comparison with previously proposed device-independent secret sharing schemes.
 Each party has black box access denoted by $BB_{_1}^A$, $BB_{_2}^B$, $BB_{_3}^C$ respectively for Alice, Bob and Charlie [Fig.~\ref{fig:Pseudo-telepathy_game}]. The devices obey the laws of quantum mechanics and are specially isolated from each other and from any adversary.
At the outset, a third party distributes $N$ copies of the GHZ state 
$\ket{\mathrm{GHZ}} = \tfrac{1}{\sqrt{2}}(\ket{000} + \ket{111})_{ABC}$, 
with each of the particles received by Alice, Bob, and Charlie, respectively.
Subsequently, each party chooses their respective input classical bit $x_1, x_2, x_3 \xleftarrow{\$} \{0,1\}$. 
 After choosing their inputs, the parties are isolated and proceed according to the pseudo-telepathy game, ultimately producing outputs $y_1, y_2, y_3$. \\
Following the measurement process, both parties publicly disclose their inputs and retain only those rounds in which the number of ones is even. A random subset of these rounds is then selected to verify whether the pseudo-telepathy winning condition is satisfied. If the observed success rate falls below a predetermined threshold, the protocol is terminated.

\emph{A General Approach.}$-$
For \( n = 3 \), only one out of four cases satisfies both device-independence verification and key generation (see Table~II in Appendix~\ref{Insufficiency of the $3$-Qubit GHZ State for Simultaneous Key Generation and DI Checking}), yielding a success probability of \( 0.25 \), which is not enough. This raises a natural problem: identifying the $n$, for which using an $n$-qubit GHZ state, with input strings $ x_1x_2\dots x_n $, produces outputs $ y_1y_2\dots y_n$ that simultaneously satisfy the pseudo-telepathy game condition~(\ref{Winning condition}) and the key generation condition $ y_1 = y_2 \oplus \dots \oplus y_n $. Furthermore, can such configurations outperform the $3$-qubit GHZ-based case?

 \emph{Theorem $1$.}\label{cond1}
Let $x_i$ denote the input bit of the $i^{th}$ participant.  
The Hamming weight of the input string $x_1x_2 \dots x_n$ is a multiple of $4$  
if and only if the following two conditions on the corresponding output string $y_1 y_2 \dots y_n$ of the pseudo-telepathy game hold: (i) it has an even number of $1$'s, and ({ii}) it satisfies the key generation condition simultaneously.\\
\noindent The proof of Theorem 1 is provided in Appendix~\ref{Proof of {Theorem $1$}}.

So, the key-generation condition is satisfied by input strings whose Hamming weight is a multiple of four, allowing the verification to be confined to this subset without loss of correctness.

The next consideration concerns the total number of valid input-output pairs in the pseudo-telepathy game. For an $n$-party scheme, there exist $2^{n-1}$ input strings $(x_1x_2\dots x_n)$ with an even number of ones, as required by the game. Among these, we focus on those whose Hamming weight is $4k$, \emph{i.e.}, a multiple of four. Hence, the goal is to
count all input-output pairs $(x, y)$ with $x, y \in \{0,1\}^n$ such that:
\begin{enumerate}
    \item The input $x$ has the Hamming weight that is a multiple of $4$, i.e.,
    \[
        \mathrm{wt}(x) \equiv 0 \pmod{4}.
    \]
    \item The output $y$ satisfies
    \[
        \sum_{i=1}^n y_i \equiv \frac{1}{2} \sum_{i=1}^n x_i \pmod{2}.
    \]
    \item The output satisfies the key generation condition.
    \[ y_1 = y_2 \oplus y_3 \oplus \dots \oplus y_n .\]
\end{enumerate}

\noindent\text{Input:} An integer $n > 0$.  

\noindent\text{Output:}  
\begin{enumerate}
    \item $\text{pair\_count}$ = number of valid $(x,y)$ pairs.
    \item $\text{ratio} = \dfrac{\text{pair\_count}}{2^{2(n-1)}}$.\\
\end{enumerate}

\begin{table}[t]
  \label{alg:math-count}
\hrule
\hrule
  \begin{tabular}{c}
     \centering
  \begin{minipage}{\columnwidth}
   \medskip
    \begin{enumerate}
      \item Let $n>0$ be an integer.
      \item Define the input and output sets:
        \[
          \mathcal{X}=\{\,x\in\{0,1\}^n:\mathrm{wt}(x)\equiv 0\pmod{4}\,\},\qquad
          \mathcal{Y}=\{0,1\}^n,
        \]
        where $\mathrm{wt}(x)$ denotes the Hamming weight of $x$.
      \item For each $x\in\mathcal{X}$ compute
        \[
          s_x=\sum_i x_i,\qquad
          \mathrm{res}_x=\Big(\frac{s_x}{2}\bmod 2\Big).
        \]
      \item For each $y\in\mathcal{Y}$ compute
        \[
          \mathrm{res}_y=\sum_i y_i \pmod 2.
        \]
      \item The total number of valid input-output pairs is
        \[
          \mathrm{pair\_count}=\sum_{x\in\mathcal{X}}\sum_{y\in\mathcal{Y}}
          \delta_{\mathrm{res}_x,\mathrm{res}_y},
        \]
        where $\delta_{a,b}$ is the Kronecker delta. $\mathrm{res}_x = \mathrm{res}_y$ is the valid pair.
      \item The corresponding ratio is
        \[
          \mathrm{ratio}=\frac{\mathrm{pair\_count}}{2^{2(n-1)}}.
        \]
    \end{enumerate}

  \end{minipage}
  \medskip
\end{tabular}
\hrule
\hrule
   \caption{Mathematical formulation for counting valid input--output pairs for $n$ parties.}
\end{table}

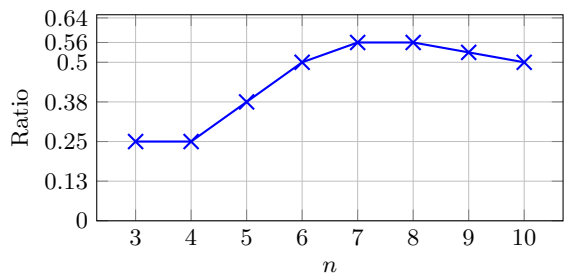
\begin{figure}[h!]
\begin{tikzpicture}
\begin{axis}[
    width=0.9\linewidth,
    height=0.5\linewidth,
    xlabel={$n$},
    ylabel={Ratio},
    ymin=0,
    ymax=0.65,
    xtick={3,4,5,6,7,8,9,10},
    ytick={0,0.125,0.25,0.375,0.5,0.5625,0.64},
    grid=both,
    grid style={line width=.1pt, draw=gray!30},
    major grid style={line width=.2pt,draw=gray!50},
    mark options={scale=2},
]
\addplot[
    color=blue,
    mark=x,
    thick
] coordinates {
    (3,0.25)
    (4,0.25)
    (5,0.375)
    (6,0.5)
    (7,0.5625)
    (8,0.5625)
    (9,0.53125)
    (10,0.5)
};
\end{axis}
\end{tikzpicture}
\caption{Variation of the ratio with the number of parties $n$. For $n = 3$, the ratio is $0.25$.  For $n = 4$, the ratio remains $0.25$.  For $n = 5$, the ratio increases to $0.375$. For $n = 6$, it further rises to $0.5$. For $n = 7$, the ratio reaches $0.5625$, and it remains the same for $n = 8$. For $n = 9$, the ratio slightly decreases to $0.53125$, and for $n = 10$, it returns to $0.5$.}
\label{fig:ratio_plot}
\end{figure}
The ratio values for the different integers $n$ reveal a notable pattern as shown in Fig.~\ref{fig:ratio_plot}.  

From the results, we observe that the ratio reaches its maximum at $n=7$, and then gradually decreases, reaching $0.5$ at $n=10$.
It is natural to ask whether the ratio can exceed $0.5625$ for any $n$, or if the maximum truly occurs at $n=7$. \\
First consider the function $f(n)$ which takes the input $n$, and outputs the ratio.
\begin{center}
    ratio $= \frac{\text{input-output pairs satisfying both the conditions}}{\text{total valid input-output pairs}}.$
\end{center}
As it is clear, for each input of size $n$, there are a total of $2^{\, (n-1)}$ possible outputs. Then we just consider the input pairs to calculate the ratio. The ratio can be written as

\begin{center}
    ratio $= \frac{\text{input with Hamming weight as a multiple of $4$}}{\text{number of valid inputs satisfying Pseudo-telepathy game}}$.
\end{center}
The inputs having the Hamming weight as the multiple of $4$ can be calculated as $\binom{n}{0} + \binom{n}{4} + \binom{n}{8} + \cdots + \binom{n}{4k}, \quad  \text{where} \; 4k \leq n . \\
$
The idea behind it is to choose the positions that are multiples of $4$ and fill them with $1$'s, while the remaining positions are filled with $0$'s.

 \emph{Theorem $2$}. Suppose $f(n)$ is a function that takes an input $n$ and outputs the corresponding ratio of the {\text{input with Hamming weight as multiple of $4$}} and the number of valid inputs satisfying the pseudo-telepathy game. Then $f(n)$ attains its maximum at $n=7$.

\noindent The proof of Theorem 2 is provided in Appendix~\ref{Proof of {Theorem $2$}}.\\
Hereafter, we consider $n=7$ and examine whether it enables a more efficient three-party scheme satisfying both pseudo-telepathy and key-generation requirements.

\emph{DI-QSS by Pseudo-telepathy game using Seven Qubit.}$-$
Now, we aim to propose a $(3,3)$ DI-QSS scheme with the underlying idea that, to design such a scheme, the inputs must be distributed among the participants in a way that ensures two essential properties: (i) the distribution remains equally likely for all parties, and (ii) the scheme produces the desired outcome consistent with the original secret sharing construction.
The required GHZ state for that is \begin{align*} 
\ket{GHZ^+} = \frac{1}{\sqrt{2}}\left[\ket{0}^{\otimes 7} + \ket{1}^{\otimes 7}\right] ,
\end{align*} 
and now to explore how this setting can be adapted for a $(3,3)$ DI-QSS scheme. For a $7$-qubit GHZ state, the parties need $7$ classical input bits to play the pseudo-telepathy game. In particular, our work is to divide the $7$ input bits into shares among the three participants in such a way that each party receives an equal portion of information while still preserving the desired relation for key generation. 

Next, we characterize how to partition the seven-qubit system to ensure each participant receives equal information.\\
{\emph{Theorem 3.}  
If a 7-qubit state is distributed among the three parties Alice, Bob, and Charlie in the partition \((1, j, 6-j)\), then the output views of the two participants, Bob and Charlie are indistinguishable when \(j = 3\).
 }\\
 The detailed proof of this is presented in Appendix~\ref{Proof of Theorem 3}.

The protocol proceeds as follows. Each participant independently selects input bits and feeds them into their corresponding device to obtain the outputs $y_i$. Specifically, Alice chooses $x_1$, Bob selects three bits $x_2, x_3, x_4$, and Charlie selects three bits $x_5, x_6, x_7$, all uniformly at random. We call the input string here
    $X = x_1\; x_2 x_3 x_4\;x_5 x_6 x_7$. 

After playing the pseudo-telepathy game, the participants publicly announce their input bits and verify that the total number of ones among $\{x_i\}$ is a multiple of 4. If this parity condition is satisfied, they check the input-output relation~(\ref{Winning condition}).

When this condition holds, the outputs are used to generate the shared key. The secret key component for Alice is $K_A = y_1$, for Bob it is $K_B = y_2 \oplus y_3 \oplus y_4,$ and for Charlie is $ K_C = y_5 \oplus y_6 \oplus y_7. $

Table {III} in the Appendix~\ref{Table 3: Updated design of DI-QSS scheme using $7$-bit input } shows that all inputs with a Hamming weight that is a multiple of $4$ satisfy the key generation condition in accordance with \emph{Theorem $1$}.

The final version of the device-independent quantum secret sharing (DI-QSS) protocol is as follows.

\noindent\hrulefill
\vspace{-2mm}\begin{enumerate}[leftmargin=1.84cm]
\item[Protocol 1.] Pseudo-telepathy game based DI-QSS protocol.
\end{enumerate}
\vspace{-4mm}\hrulefill
\begin{enumerate}

  \item \textbf{Setup.} Start with $R$ copies of the 7-qubit GHZ state
  \[
    \ket{\mathrm{GHZ}^{+}} = \frac{1}{\sqrt{2}}\big(\ket{0}^{\otimes 7} + \ket{1}^{\otimes 7}\big),
  \]
  and distribute one qubit to Alice and three qubits each to Bob and Charlie.
  
  \item \textbf{Round $i=1, 2, \dots, R$.} The parties independently choose inputs
  \begin{align*}
  \hspace{6mm}  x_1 \in \{0,1\}, \; x_2,x_3,x_4 \in \{0,1\}, \; x_5,x_6,x_7 \in \{0,1\}, \end{align*}
  use them on their (untrusted) devices, and obtain outputs $y_1,\dots,y_7 \in \{0,1\}$.  
  For the $j^{\text{th}}$ qubit, the device applies the unitaries $S$ and $H$, where 
  $$S = \begin{bmatrix}
      1 & 0\\
      0 & e^{\frac{i \pi}{2}}
  \end{bmatrix} \; \text{and} \; H = \frac{1}{\sqrt{2}}\begin{bmatrix}
      1 & 1\\
      1 & -1
  \end{bmatrix},$$ then measures in the computational basis $\{\ket{0},\ket{1}\}$.  
  
  \item \textbf{Sifting.} The parties publish their inputs for all rounds and retain only those rounds whose input Hamming weight is a multiple of $4$. Denote the set of retained rounds by $\mathcal{A}$.
  
  \item \textbf{Test.} Alice randomly selects a subset $\mathcal{B} \subset \mathcal{A}$ of size $\gamma|\mathcal{A}|$ ($0<\gamma<1$) and announces the input–output pairs for rounds in $\mathcal{B}$.  
  Bob and Charlie do likewise. They compute the empirical success probability of the pseudo-telepathy parity condition:
  \[
    \sum_{j=1}^7 y_j \equiv \tfrac{1}{2}\sum_{j=1}^7 x_j \pmod{2}.
  \]
  If the observed success probability is less than $1-\eta$, they abort.
  
  \item \textbf{Key generation.} From the remaining rounds $\mathcal{A}\setminus\mathcal{B}$, each kept round yields raw key bits:
  \begin{align*}
  \hspace{6mm}  K_A = y_1, \; 
    K_B = y_2 \oplus y_3 \oplus y_4, \;
    K_C = y_5 \oplus y_6 \oplus y_7,
  \end{align*}
  which satisfy $K_A = K_B \oplus K_C$ in ideal runs.
\end{enumerate}
\hrulefill


\emph{Correctness analysis.}$-$
After successfully passing the checking phase (i.e., satisfying the winning condition of the pseudo-telepathy game), 
Alice, Bob, and Charlie each obtain their respective measurement outcomes by measuring in the computational basis 
$\{\ket{0}, \ket{1}\}$. According to the protocol, let the classical outputs be $y_1$, $y_2 y_3 y_4$, and $y_5 y_6 y_7$ 
for Alice, Bob, and Charlie, respectively.
\noindent These outputs should satisfy the key generation condition
   \begin{equation}
       K_A = K_B \oplus K_C    \label{key_gen_cond} .
   \end{equation}  
   
Thus, from the sharing of the GHZ state up to the reconstruction phase, via the steps of the proposed scheme, 
the fact that the secret is perfectly recovered demonstrates the correctness of the scheme. 
We restate the correctness as per the protocol Mukhopadhyay et al.~\cite{PhysRevA.100.012319} as follows.

\textbf{Definition 1. (Correctness).}
For $\epsilon_{\mathrm{cor}} > 0$, a quantum secret sharing protocol is said to be 
\(\epsilon_{\mathrm{cor}}\)-correct, if for any adversary, the original secret \(S\) 
and the reconstructed secret \(\hat{S}\) are \(\epsilon_{\mathrm{cor}}\)-indistinguishable, i.e.,
    $\Pr_{S \in \mathcal{S}}[\, S = \hat{S} \,] \geq 1 - \epsilon_{\mathrm{cor}}$,
where $\mathcal{S}$ is the set of all possible secrets.

As mentioned earlier, if the fraction of successful rounds falls below $1-\nu$, the dealer aborts the protocol. The dealer fixes the length of the key in advance, and this key length serves as a parameter in defining the correctness of the scheme. \\
\emph{Theorem 4.}
The proposed device-independent secret-sharing scheme is $\varepsilon_{\mathrm{cor}}$-correct, where
$\varepsilon_{\mathrm{cor}} > 1 - X$, and
\[
X = \binom{R_{l}}{(1-\nu)R_{l}}\, p_{m}^{(1-\nu)R_{l}} (1-p_{m})^{\nu R_{l}} ,
\]
where
\begin{itemize}
    \item $R_{l}$ denotes the number of rounds fixed by the dealer to obtain a final key of length~$l$,
  \item $\nu$ denotes the allowed proportion of deviation, and 
    \item $p_{m}$ is the probability that Alice’s bit matches Bob${}\oplus{}$Charlie’s output after one round of the protocol.
 \end{itemize}

\noindent See Appendix~\ref{Proof of Theorem 4} for the proof of \emph{Theorem 4}.


\emph{Security Analysis}.$-$
We examine adversarial attacks that are independent and identically distributed (IID) to determine the protocol's security. 

\textbf{Definition 2. (Security against collective attack model).}
A DI-QSS protocol involving $3$ parties A, B, and C is called secure against the collective attack model if the following holds
$$\inf_{\rho_E \in \Gamma} [H(A | E)_{\rho_E} - H(A |B, C)_{\rho_E}] > 0,$$
where $\Gamma$ is the set of states consistent with the observed experimental statistics, and $H(\cdot |\cdot )$ is the conditional (Von-Neumann) entropy that depends on the state $\rho_E$ of the eavesdropper.


In this collective attack model, it is possible to assume that the adversary Eve prepares a quantum state $\rho_{ABCE}$~\cite{ArnonFriedman2018, pironio2009device} in each round of the protocol and gives Alice, Bob and Charlie.
Let $\mathcal{H}_A$, $\mathcal{H}_B$, $\mathcal{H}_C$, and $\mathcal{H}_E$ denote the separable Hilbert spaces corresponding to Alice, Bob, Charlie, and Eve, respectively. Now consider the scenario in which, at the start of each round, a shared general quantum state $\rho_{ABCE}$ is distributed among the four systems. Upon receiving their respective inputs $x$, the devices of Alice, Bob, and Charlie perform measurements described by predefined POVM's $\{M_{y_1|x_1}\}_{A}$, $\{N_{y_2y_3y_4|x_2x_3x_4}\}_{B}$, and $\{P_{y_5y_6y_7|x_5x_6x_7}\}_{C}$ (where, {$M_{y|x}$ denotes the measurement done by Alice when the input bit is $ x$.}), producing corresponding outputs. The exact forms of these states and measurements are unknown to the users but may be fully accessible to Eve. The framework of device-independent security is designed precisely to address such scenarios, ensuring security even when the measurement devices are under Eve’s control.



The joint conditional distribution  of their outputs can be written as
\begin{align*}
    &p(y_1, y_2y_3y_4, y_5y_6y_7 | x_1, x_2x_3x_4, x_5x_6x_7)\\
    &= Tr\left[\rho_{ABCE}(M_{A} \otimes N_{B} \otimes P_{C} \otimes I_{E}) \right].
\end{align*}

Here, $X=x_1x_2 \dots x_7$ is the input randomly selected by Alice Bob, and Charlie for their respected measurement devices. ${{y_1}}$ denotes the key bit generated from Alice’s input $x_1$ with analogous definitions for Bob and Charlie. After the completion of one round Eve's quantum system is described by  
$$\rho_{E}^{yx} = Tr_{ABC}\left[\rho_{ABCE}(M_{A}\otimes N_{B} \otimes P_{C} \otimes I_{E}) \right].$$
$\rho_{E}^{yx}$ represents Eve’s information associated with the keys of the three parties. \\
So the key generation rate of the DI-QSS protocol can thus be expressed as follows 
\begin{align}
   \notag r &= I(A;B,C) - I(A;E)\\
   \notag &= H(A)-H(A|B,C)-H(A)+H(A|E)\\
    &= H(A|E)_{\rho_{E}} - H(A|B,C)_{\rho_{E}}. \label{(DWbound)}
\end{align}
Eq.~(\ref{(DWbound)}) is known as the Devetak-Winter bound~\cite{Devetak2005},
that is a universal method for calculating the key
rate in the quantum cryptography field, which has been widely
used in QKD and QSS systems~\cite{Qi_2021, PhysRevA.80.042307}.
The first term bounds Eve’s information, and the second captures the error rate.

The security of the protocol must account for all possible  ${\rho}_{ABCE}$, implying that the final key is determined by the worst-case scenario over all possible quantum strategies. To demonstrate that the protocol can successfully generate correct and secure keys, it must be shown that the key rate remains positive when the protocol is implemented using those quantum strategies.
In the ideal scenario, the protocol begins with the quantum state $\ket{GHZ^+} = \frac{1}{\sqrt{2}}(\ket{0}^{\otimes 7} + \ket{1}^{\otimes 7})$. Since the winning probability of the pseudotelepathy game equals one for valid input pairs, the dealer sets the error parameter $\eta$ to zero. The inputs selected by the parties are independent and random, and as long as they adhere to the prescribed steps of the protocol, it will not be aborted. Moreover, $H(A|B,C) = 0$ for all valid input combinations (see Appendix~\ref{$H(A|B,C) = 0$ for three parties}). Given the output bits of Bob and Charlie, there is no uncertainty in determining Alice’s output bit. To achieve the key rate fully i.e., $r = 1$, it is therefore necessary to show that $H(A|E) = 1$. {The GHZ state is a maximally entangled multipartite quantum state shared among the legitimate parties. Consequently, it cannot be correlated with the eavesdropper, who constitutes the fourth subsystem.} 

The non-ideal case considers scenarios in which arbitrary quantum strategies may be employed. In this setting, the key rate $r$ can be bounded for all valid input–output pairs $(x, y)$. To establish a lower bound on the key rate $r$, it suffices to bound the conditional entropy $H(A|E)_{\rho_{E}}$. The recently developed quasirelative entropy technique~\cite{Brown2024deviceindependent} provides a means to derive such a lower bound on $H(A|E)_{\rho_{E}}$. A precise lower bound for $H(A|E)_{\rho_{E}}$ is derived in~\cite{Brown2024deviceindependent} as 

\begin{align}
H(A|E) \geq\; & c_m 
+ \sum_{i=1}^{m-1} 
  c_i
  \sum_{y \in \{0,1\}} 
  \inf_{Z_y \in E} 
  \Tr\Big\{
     \rho_{AE}
     \Big[
       M_{y|x} \otimes 
       \Gamma
       \nonumber \\[2pt]
     & \hspace{7em}
       +\, t_i 
       \big(
         I_A \otimes Z_y Z_y^{\dagger}
       \big)
     \Big]
  \Big\},
\end{align}
 where $\Gamma = \big(
         Z_y + Z_y^{\dagger} 
         + (1 - t_i) Z_y^{\dagger} Z_y
       \big)$,
    $ c_m = \sum_{i=1}^{m-1}c_i$,  and  $c_i = \frac{w_i}{t_i(ln 2)}.$
$M_{y|x}$ denotes the measurement done by Alice when the input bit is $ x$. \\
$\rho_{ABCE}$ initial quantum state, so $\rho_{AE} = \Tr_{BC}(\rho_{ABCE})$.\\
Also, $\{(t_i,w_i)|i=1,2, \dots m\} $ a set of $m$ nodes and weights of the Gauss-Radau quadrature. $Z_y$ is an arbitrary operator. This minimization problem can be formulated as a semidefinite programming problem and solved to obtain the $Z_y$. Thus, the result implies that $H(A|E)$ is always bounded below by a positive quantity, whatever be the specific states employed by Eve through the state $\rho_{ABCE}$. Consequently, in the noiseless scenario, the key rate $r$ remains strictly positive, indicating that the protocol is capable of generating secure keys in the presence of collective attacks.

\emph{Performance of the DI-QSS under noisy condition.$-$} Over long distance photon transmission channels, photon loss, and decoherence caused by channel noise can significantly destroy entanglement and weaken the non-local correlations among the users’ measurement results during practical quantum communication. 

We adopt the white-noise model, which is standard also in device-independent QKD protocols~\cite{woodhead2021device, masini2022simple, sekatski2021device, PhysRevLett.131.080801}. 
Within this model, the ideal $7$-qubit GHZ state is transformed into a uniform mixture over the full set of $2^{7}=128$ GHZ-states. We further assume a global detection efficiency $\eta$, so that the probability of a no-click (non-detection) event is $1-\eta = \bar{\eta}$.
Finally, the mixed state of Alice, Bob and Charlie for a $7$-qubit state, including all noise effects is 
\begin{widetext}
\begin{align}
\rho_{ABC}
&= \eta^{7}\!\Big(
F\,|GHZ^{+}\rangle\langle GHZ^{+}|
+\frac{1-F}{2^{7}}\, \mathbb{I}_{2^{7}}
\Big) \nonumber\\
&\quad + \sum_{k=1}^6 \binom{7}{k}\eta^{(7-k)}(\bar{\eta})^k \frac{1}{2}\Big(
| 0^{\otimes (7-k)}\rangle\langle 0^{\otimes (7-k)}| 
+| 1^{\otimes (7-k)}\rangle\langle  1^{\otimes (7-k)}|
\Big) \nonumber \\
&\quad + \bar{\eta}^{7}\,
|{\rm vac}\rangle\langle{\rm vac}|. 
\end{align}
\end{widetext}
where $F$ denotes the probability that the photon state is error-free,
$\mathbb{I}_{2^{7}}$ represents the identity operator corresponding to the
uniform mixture over all $2^{7}$ seven-qubit GHZ basis states, and $\ket{\mathrm{vac}}$ denotes the vacuum state. According to the multiparty pseudotelepathy game, all $\ket{GHZ^{-}_{j}}$ states (The full expressions can be found in Appendix~\ref{Mixture of the GHZ states})
are not fit for the game because the final output of those states consists of an odd number of $1$'s~\cite{10.1007/978-3-540-45078-8_1}. Therefore, for the decoherence, the QBER is 
        \begin{equation}Q_1 = 2^6 \left(\frac{1-F}{2^7}\right)\eta^7 = \left(\frac{1-F}{2}\right)\eta^7 . \end{equation}
We now evaluate the effect of photon loss on the DI-QSS protocol, noting that valid outputs always contain an even number of 
1's. We analyze the cases sequentially: no photon loss, single-photon loss, two-photon loss, and finally complete photon loss (See Appendix~\ref{QBER calculation} for detailed calculations). Combining all these, the QBER $Q_2$ for photon loss is 
\begin{align*} Q_2 =  \frac{1}{2}(1- \eta^7) . \end{align*}
where the $\frac{1}{2}$ implies the error probability. Since the key generation condition is a single binary constraint, a uniformly random bit (after photon loss) violates it with probability $\frac12$.

Now the total bit error rate combining the decoherence and photon loss is 
\begin{align} Q = Q_1 + Q_2 &= \left(\frac{1-F}{2}\right)\eta^7 + \frac{1}{2}(1-\eta^7) \notag \\
&=  \frac{1}{2} - \frac{F}{2}\eta^7 .\end{align}
Therefore, in the noisy model $H(A|B,C)$ turned into~\cite{acin2007device} 
\begin{equation*}  H(A|B,C) = h(Q). \end{equation*}
Previously, in the ideal scenario, we had $H(A|B,C) = h(0) = 0$. After introducing errors, this value becomes $h(Q)$, where $h(\cdot)$ denotes the binary entropy function.

Next, we establish a lower bound on $H(A|E)$ to ensure a positive key rate under noisy conditions, as follows from the Eq.~(\ref{(DWbound)}).

The obtained bound on $H(A|E)$ is as follows. 
\begin{align*}
 H(A | E) \geq  1 - h(Q).
\end{align*}
(See Appendix~\ref{Lower Bound of $H(A|E)$ under Noisy Scenario} for detailed calculations).

Substituting the derived entropy bound into the Devetak-Winter relation, the asymptotic secret key rate $r$ is obtained as follows:
 \begin{align*}
     r \geq 1 - h(\tfrac12 - \tfrac12 \eta^{7}F) - h(\tfrac12 - \tfrac12 \eta^{7}F) .
 \end{align*}
 Now from this to get a positive key rate $r$, we consider two cases\\
Case 1: When $F=1$ there is no decoherence, 
\begin{align}
  r \geq  1 - 2\,h(\tfrac12 - \tfrac12 \eta^{7}F) .
\end{align}
Then for $\eta \geq 0.97$ it gives always a positive key rate.\\
Case 2: When there is some decoherence \emph{i.e.} $F\neq 1$. 
The secret key rate, defined by $r \geq 1 - 2h\left(\frac{1}{2}  -\frac{\eta^{7}F}{2}\right)$, remains positive only when the condition $\eta^{7}F > 0.78$ is satisfied. This requirement establishes the transmission efficiency is $\eta > (0.78/F)^{1/7}$. Optimal operation of the protocol is achieved at a fidelity of $F = 0.95$, requiring a minimum threshold efficiency of $\eta > 0.972$ to ensure a non-zero key rate.

\emph{Bitwise improvement over the CHSH-based DI-QSS scheme.}$-$ Initially, let both schemes be executed over $R$ rounds to obtain the final result. 
\begin{equation}
\begin{aligned}
\text{Alice: } \; \mathcal{A}_{1} &= \sigma_{x}, &  \mathcal{A}_{2} &= \sigma_{y};  \\
\text{Charlie: } \; \mathcal{C}_{1} &= \sigma_{x}, & \; \mathcal{C}_{2} &= -\sigma_{y}; \\
\text{Bob:} \; \mathcal{B}_{1} &= \sigma_{x}, & \; 
\mathcal{B}_{2} &= \tfrac{\sigma_{x} - \sigma_{y}}{\sqrt{2}}, & \; 
\mathcal{B}_{3} &= \tfrac{\sigma_{x} + \sigma_{y}}{\sqrt{2}}. 
\end{aligned}
\label{CHSH basis}
\end{equation}

In the earlier approach Zhang \emph{et al.}~\cite{PhysRevA.110.042403} proposed the measurement bases as of Eq.~\ref{CHSH basis} ({See Appendix~\ref{DI-QSS Protocol of Zhang}  for detailed steps of the protocol}). 

According to the protocol, $R$ rounds require a total $R$ GHZ states in order to complete all the steps. Participants choose measurement bases independently and uniformly at random. Consequently, $\tfrac{1}{12}$ of the rounds contribute to key generation, while the remainder are utilized for device-independence (DI) verification or discarded. We assign $\alpha R$ rounds (where $0 < \alpha < 1$) to check CHSH polynomial, reserving the remaining $(1-\alpha)R$ states for secret key extraction.
\begin{align*}
\tfrac{1}{12}(R-\alpha R) :& \; \text{ the expected number of cases for}\\
                            & \text{key generation.} 
\end{align*}
Now, for our protocol, with $R$ rounds, only $\tfrac{R}{2}$ of them correspond to inputs with an even number of $1$’s. Among these $\tfrac{R}{2}$ states, a fraction $\beta$ (with $0 < \beta < 1$) is reserved for the checking phase, i.e., $ \frac{\beta R}{2}$ states are for checking the device independence. The remaining $\frac{(R - \beta R)}{2}$ states are used for the key generation. Furthermore, the updated pseudo-telepathy game ensures that $56.25\%$ of these cases are valid simultaneously for both satisfying the pseudo-telepathy winning condition and for key generation. 
\begin{align*}
\tfrac{0.5625}{2}(R-\beta R):& \;
\text{expected number of cases for}\\
& \text{key generation condition.}
\end{align*}
Now if we can choose the fraction $\beta $ as close as the previous fraction $\alpha$ then those two factors $R-\alpha R$ is very close to the $R - \beta R$. Hence we can compare these two cases easily. Therefore, the advantage is 
$$\frac{\frac{0.5625}{2}(R-\beta R)}{\frac{1}{12}(R - \alpha R)} = 3.375 \left(\frac{1-\beta }{1 - \alpha }\right).$$
$\alpha \approx \beta$ will imply the advantage will be near $3.375$.

\emph{Discussion$-$} In this work, we propose a simple device-independent quantum secret sharing scheme based on the multiparty pseudo-telepathy game, where the protocol lies in the fact that it does not require different measurement bases; a single test suffices for both device-independence verification and the security of the generated key bits. In contrast to earlier secret-sharing schemes based on the Svetlichny inequality, we introduce a three-party protocol using seven-qubit GHZ states such that the construction ensures that the same test is sufficient for verifying device independence as well as for the successful regeneration of the secret key bits. Moreover, we have shown that our protocol remains secure against collective attacks, and also produces a positive asymptotic secret key rate under the white noise model. It also provides certain bitwise advantages over previously proposed device-independent protocols. 
As a future direction, we plan to expand the analysis under a more realistic communication scenario that simultaneously accounts for collective attacks by an adversary and the presence of noise.
Additionally, another future research will focus on developing sophisticated methods for enhancing noise robustness beyond the three-party setting.





\clearpage

\appendix

\setcounter{secnumdepth}{2}

\section{Analysis of Hillery et al.~\cite{PhysRevA.59.1829} scheme} \label{Analysis of Hillery et al.}
   The first QSS protocol based on the Greenberger-Horne-Zeilinger(GHZ) state,
\[
\ket{\Psi}=\frac{\ket{000}+\ket{111}}{\sqrt{2}},
\]
shared among Alice, Bob, and Charlie was introduced by Hillery et al.~\cite{PhysRevA.59.1829} and it can be seen as an extension of the quantum key distribution (QKD)~\cite{BB84}, where the secret key is shared between the two parties. Each party holds one qubit of the triplet and randomly chooses to measure in either the $\sigma_x$ or $\sigma_y$ basis, with measurement settings corresponding to the eigenstate
$\sigma_x = \{\ket{+},\ket{-}\},\; \sigma_y = \{\ket{+i},\ket{-i}\}$.

After performing their measurements, they publicly announce only the chosen bases. In roughly half of the rounds~[Fig.~\ref{Hillery's basis}], the correlations of the GHZ state allow Bob and Charlie, by combining their outcomes, to infer Alice’s result. This property enables Alice to establish a shared secret key with Bob and Charlie, thereby laying the foundation of quantum secret sharing.
The GHZ state \( \ket{\Psi} \) can be rewritten in terms of \( \ket{+} \) and \( \ket{-} \) as  
\begin{equation}
\begin{aligned}
\ket{\Psi} = &\frac{1}{2\sqrt{2}} \big[ \,
 {\left( \ket{+}_A\ket{+}_B + \ket{-}_A\ket{-}_B \right)} (\ket{0} + \ket{1})_C  \\
 &+ {\left( \ket{+}_A\ket{-}_B + \ket{-}_A\ket{+}_B \right)} (\ket{0} - \ket{1})_C
\big].
\end{aligned}
\label{hillery"s basis}
\end{equation}

\begin{figure}[h]
\centering
\scalebox{1}{
\setlength{\arrayrulewidth}{0.34mm}
\renewcommand{\arraystretch}{1.9}
\begin{tabular}{|c|c|c|c|c|}
\hline
\multicolumn{1}{|c|}{\multirow{1}{*}{\begin{tabular}{c}Basis\end{tabular}}} 
 & \multicolumn{4}{c|}{Alice} \\ \cline{1-5}
 Bob & \multicolumn{1}{c}{\color{red}{$\lvert{+x}\rangle_a$}} & \multicolumn{1}{c}{\color{red}{$\lvert{-x}\rangle_a$}} & \multicolumn{1}{c}{\color{red}{$\lvert{+y}\rangle_a$}} & \color{red}{$\lvert{-y}\rangle_a$} \\ \hline
\multirow{2}{*}{\begin{tabular}{c} \color{darkgreen}{$\lvert{+x}\rangle_b$} \\ \color{darkgreen}{$\lvert{-x}\rangle_b$} \end{tabular}} 
 & \color{darkgray}{$\lvert{+x}\rangle_c$} & \color{darkgray}{$\lvert{-x}\rangle_c$} & \color{darkgray}{$\lvert{+y}\rangle_c$} & \color{darkgray}{$\lvert{-y}\rangle_c$} \\ \cline{2-5}
 & \color{darkgray}{$\lvert{-x}\rangle_c$} & \color{darkgray}{$\lvert{+x}\rangle_c$} & \color{darkgray}{$\lvert{-y}\rangle_c$} & \color{darkgray}{$\lvert{+y}\rangle_c$}\\ \hline
\end{tabular} }
\caption{ \label{Hillery's basis}
Correlation structure among the three users’ measurement results. 
Alice’s measurement bases are listed in the top row, Bob’s in the left column, 
and Charlie’s corresponding outcomes are shown in the table cells. 
Note that Charlie’s result alone does not determine the individual outcomes of Alice and Bob.%
}
\end{figure}
From $\ket{\psi}$, Charlie can infer whether their outcomes are correlated or anti-correlated.
The probability of either correlation or anti-correlation is \( \frac{1}{2} \). Therefore, Charlie gains no advantage and has no option other than to make a random guess.

\section{DI-QSS Protocol of Zhang \emph{et al.}\cite{PhysRevA.110.042403} } \label{DI-QSS Protocol of Zhang}
The protocol proceeds in five stages under the assumption that all three players must be legitimate and honest. 

\textbf{Step $\mathsf{1}$}: A central source prepares $n$ copies of tripartite GHZ states, distributing the three photons of each state to Alice, Bob, and Charlie, respectively for each copy. Each party thus receives a sequence of qubits, denoted $E_1$, $E_2$, and $E_3$.

\textbf{Step $\mathsf{2}$}: The users independently and randomly choose the measurement bases. 
\begin{equation}
\begin{aligned}
\text{Alice:} \quad \mathcal{A}_{1} &= \sigma_{x}, & \, \mathcal{A}_{2} &= \sigma_{y},  \\
\text{Charlie:} \quad \mathcal{C}_{1} &= \sigma_{x}, & \, \mathcal{C}_{2} &= -\sigma_{y}, \\
\text{Bob:} \quad \mathcal{B}_{1} &= \sigma_{x}, & \, 
\mathcal{B}_{2} &= \tfrac{\sigma_{x} - \sigma_{y}}{\sqrt{2}}, & \, 
\mathcal{B}_{3} &= \tfrac{\sigma_{x} + \sigma_{y}}{\sqrt{2}}. 
\end{aligned}
\label{CHSH basis}
\end{equation}
The measurement outcomes are recorded as $a_j, b_j, c_j \in \{+1, -1\}$ for the $j$th round.

\textbf{Step $\mathsf{3}$}: Parameter estimation is performed once the measurement bases are publicly announced. Depending on Bob’s chosen basis, three cases arise.\\
\text{Case 1:} Bob chooses the basis $B_2$ or $B_3$. Total eight basis combinations are there\\ $\{ A_1B_2C_1, A_2B_2C_1,\dots, A_2B_3C_2\}$ are dedicatedly used for constructing the Svetlichny polynomial~\cite{PhysRevD.35.3066, PhysRevLett.106.020405} $S_{ABC}$ and checking device independence.\\
\text{Case 2:} When the measurement basis combination is $A_1B_1C_1$, then the users retain their measurement results as the raw key bits.\\
\text{Case 3:} For the basis combinations $A_1B_1C_2$, $A_2B_1C_1$, and $A_2B_1C_2$, 
the users discard their measurement results. 

\textbf{Step $\mathsf{4}$}: Error correction and privacy amplification are applied iteratively until a sufficient number of secure key bits are established.

\textbf{Step $\mathsf{5}$}: Finally, in the secret reconstruction phase, Bob and Charlie obtain keys $K_B$ and $K_C$, respectively, such that Alice’s key satisfies 
$K_A = K_B \oplus K_C$.

The existing protocol requires
eight basis combinations for device-independence testing and one combination for the key generation. Our aim is to design a simpler scheme that achieves the same security guarantees while consuming fewer resources. Specifically, our objective is to minimize measurement overhead without compromising correctness or security.

\section{Quantum strategy for the Pseudo-telepathy game~\cite{10.1007/978-3-540-45078-8_1}} \label{Quantum strategy for the Pseudo-telepathy game}

\emph{Quantum strategy.}
The perfect quantum protocol, in which the players always win, relies on shared entanglement. Before the game begins, the $n$ players $A_1, A_2, \dots, A_n$ share the entangled state
\[
\ket{\Phi_n^+} = \frac{1}{\sqrt{2}} \big( \ket{0^n} + \ket{1^n} \big),
\]
with each player receiving one qubit. Upon receiving the classical input $x_j$, the player $A_j$ applies the unitary
\[
S = \begin{bmatrix} 1 & 0 \\ 0 & e^{i \pi /2} \end{bmatrix}
\]
to its qubit, such that 
\[
S \ket{\beta_j} \longrightarrow e^{i \pi x_j /2} \ket{\beta_j},
\]
where $\ket{\beta_j}$ is the qubit held by $A_j$. Observe that when $\beta_j = 1$, only the operation $S$ is applied.
If $p$ players apply $S$, then
\[
\ket{\Phi_n^+} \longrightarrow 
\begin{cases}
\ket{\Phi_n^+}, & \text{if } p \equiv 0 \pmod 4,\\[1mm]
\ket{\Phi_n^-}, & \text{if } p \equiv 2 \pmod 4,
\end{cases}
\]
with $\ket{\Phi_n^-} = \frac{1}{\sqrt{2}} (\ket{0^n} - \ket{1^n})$. Cases with $p \equiv 1,3 \pmod 4$ do not occur, since the number of inputs $x_j = 1$ is always even.

Next, each player applies a Hadamard transformation $H$ to their qubit, yielding
\begin{align*}
H^{\otimes n} \ket{\Phi_n^+} = \frac{1}{\sqrt{2^{n-1}}} \sum_{{wt(y)\,} = \text{ even }} \ket{y}, \\
H^{\otimes n} \ket{\Phi_n^-} = \frac{1}{\sqrt{2^{n-1}}} \sum_{wt(y)\, = \text{ odd }} \ket{y},
\end{align*}
where $wt(y)$ denotes the Hamming weight of $y$. Finally, each player measures their qubit in the $\{\ket{0},\ket{1}\}$ basis to obtain $y_j$, and the string $y_1y_2\dots y_n$ constitutes the output corresponding to the input $x_1x_2\dots x_n$.

Therefore, if $\frac{\sum_{j}x_j}{2}$ is even, then the resulting state  
\[
\frac{1}{\sqrt{2^{n-1}}} \sum_{wt(y) \equiv 0 \pmod 2} \ket{y},
\]
indicates that the parties always succeed in winning the game based on the Hamming weight of $y$.

\section{Proof of {Theorem $1$}} \label{Proof of {Theorem $1$}}
\begin{theorem}
    Let $x_i$ denote the input bit of $i^{th}$ participant.  
The Hamming weight of the input string $x_1x_2 \dots x_n$ is a multiple of $4$  
if and only if the following two conditions on the corresponding output string $y_1 y_2 \dots y_n$ of the pseudo-telepathy game hold: (i) it has an even number of $1$'s, and ({ii}) it satisfies the key generation condition simultaneously.
\end{theorem}

\begin{proof}
Let $x_1x_2x_3 \dots x_n$ be the input bit string chosen by the participants, whose Hamming weight is a multiple of $4$. Then $x_1x_2x_3...x_n$ contains $4k$ number of $1$'s, where $k= 0,1,2,3,\dots$.
\[
\frac{1}{2} \sum x_i = \frac{1}{2} \times 4k = 2k.
\]
Hence, 
\[
\frac{1}{2} \sum x_i \pmod{2} = 0
\quad \Rightarrow \quad
\sum y_i \equiv 0 \pmod{2}.
\]
Thus, $y_i$ must have an even number of $1$'s. Now there are two cases: (i) bit $y_1$ is $1$, (ii) bit $y_1$ is $0$.

If $y_1=1$, then $y_2y_3\dots y_n$ has an odd number of $1$’s, implying $y_1 = y_2 \oplus y_3 \oplus \dots \oplus y_n = 1$. Similarly, if $y_1=0$, then $y_2y_3\dots y_n$ has an even number of $1$’s, giving $y_1 = 0 = y_2 \oplus y_3 \oplus \dots \oplus y_n = 0$.

Conversely, if $y_1y_2\dots y_n$ contains an even number of $1$'s, it satisfies $y_1 = y_2 \oplus y_3 \oplus \dots \oplus y_n$. Since we only consider pairs satisfying the pseudo-telepathy condition
\[
\frac{1}{2}\sum_i x_i = \sum_i y_i \pmod 2,
\]
it follows that $\sum_i x_i = 0 \pmod 4$. Hence, the Hamming weight of $x_1x_2 \dots x_n$ is $4k$.
\end{proof}


\begin{table*}[ht]
\centering
\renewcommand{\arraystretch}{1.1}
\begin{tabular}{ccc|ccc|c}
\toprule
\multicolumn{3}{c|}{\textbf{Input}} & 
\multicolumn{3}{c|}{\textbf{Possible Output}} & 
\textbf{Operation} \\ 
\cmidrule(lr){1-3} \cmidrule(lr){4-6} \cmidrule(lr){7-7}
\multicolumn{3}{c|}{Cond: Even no. of 1's} & 
\multicolumn{3}{c|}{$\sum y_i = \tfrac{\sum x_i}{2} \pmod{2}$} & 
$y_1 = y_2 \oplus y_3$ \\
\textbf{Alice} $(x_1)$ & \textbf{Bob} $(x_2)$ & \textbf{Charlie} $(x_3)$ & 
\textbf{Alice} $(y_1)$ & \textbf{Bob} $(y_2)$ & \textbf{Charlie} $(y_3)$ & 
  \\ 
\midrule 
0 & 0 & 0 & 0 & 0 & 0 & $0 = 0 \oplus 0$ \\
  &   &   & 0 & 1 & 1 & $0 = 1 \oplus 1$ \\
  &   &   & 1 & 0 & 1 & $0 = 0 \oplus 1$ \\
  &   &   & 1 & 1 & 0 & $0 = 1 \oplus 0$ \\
\hline
0 & 1 & 1 & 0 & 0 & 1 & $0 \neq 0 \oplus 1$ \\
  &   &   & 0 & 1 & 0 & $0 \neq 1 \oplus 0$ \\
  &   &   & 1 & 0 & 0 & $0 \neq 0 \oplus 1$ \\
  &   &   & 1 & 1 & 1 & $1 \neq 1 \oplus 1$ \\
\hline
1 & 0 & 1 & 1 & 0 & 0 & $1 \neq 0 \oplus 0$ \\
  &   &   & 0 & 1 & 0 & $0 \neq 1 \oplus 0$ \\
  &   &   & 0 & 0 & 1 & $0 \neq 0 \oplus 1$ \\
  &   &   & 1 & 1 & 1 & $1 \neq 1 \oplus 1$ \\
\hline
1 & 1 & 0 & 0 & 1 & 0 & $0 \neq 1 \oplus 0$ \\
  &   &   & 0 & 0 & 1 & $0 \neq 0 \oplus 1$ \\
  &   &   & 1 & 0 & 0 & $1 \neq 0 \oplus 0$ \\
  &   &   & 1 & 1 & 1 & $1 \neq 1 \oplus 1$ \\
\bottomrule
\end{tabular}
\caption{Strategy of DI-QSS for $n=3$}
\label{tab:3_parties_game}
\end{table*}

\section{Insufficiency of the $3$-Qubit GHZ State for Simultaneous Key Generation and DI Checking} \label{Insufficiency of the $3$-Qubit GHZ State for Simultaneous Key Generation and DI Checking}
{In Table~\ref{tab:3_parties_game}, we summarize all possible input-output combinations for the case $n=3$. According to the pseudo-telepathy game strategy, there exist four input settings for which the corresponding input-output pairs are valid for the device-independence test. However, only the input $(0,0,0)$ satisfies the conditions required for the key-generation step of the QSS protocol, as indicated in Table~\ref{tab:3_parties_game}. 

Moreover, across all the inputs, the output sequence $y_1y_2y_3$ exhibits a notable feature: Charlie’s output $y_3$ alone does not determine Alice’s and Bob’s outputs individually but merely specifies whether their results are correlated or anti-correlated, precisely the same characteristic observed in Hillery’s original protocol~\cite{PhysRevA.59.1829}. This structural similarity provides motivation that the present approach can successfully support a secure key-generation protocol.
}


\section{Proof of {Theorem $2$}} \label{Proof of {Theorem $2$}}
\begin{theorem}
    Suppose $f(n)$ is the function that takes an input $n$ and outputs the corresponding ratio of the {\text{input with Hamming weight as multiple of $4$}} and the number of valid inputs satisfying the pseudo-telepathy game. Then $f(n)$ attains its maximum at $n=7$.
Then $f(n)$ attains its maximum at $n=7$.
\end{theorem}

\begin{proof} As the function yields the aforementioned ratio of the inputs with Hamming weight multiple of $4$ and number of valid inputs satisfying Pseudo-telepathy game, for a given $n$. It suffices to show that $f(n)$ achieves its maximum precisely when the corresponding ratio is maximized.\\

\noindent Let 
\begin{align*}
f(n) \;=\; \frac{1+\binom{n}{4}+\binom{n}{8}+\cdots+\binom{n}{4k}}{2^{\,n-1}} \end{align*} \\
$ \text{where } 4k \leq n \; \text{and}\; k \in \mathbb{N}.$

This sum in the denominator is the total of binomial coefficients with indices divisible by $4$. 

For $\omega = e^\frac{{2 \pi i}}{4}$ and for the polynomial $F(x) = (1+x)^n$ if we apply the
{root of unity filter}~{\cite{riccati1757}}, then we obtain
 \begin{equation*}
\sum_{j \equiv 0 \pmod{4}} \binom{n}{j}
= \frac{1}{4}\Big[(1+1)^n + (1+i)^n 
+ (1-1)^n + (1-i)^n\Big]
\end{equation*}

Simplifying gives
\[
\sum_{j\equiv 0 \pmod{4}} \binom{n}{j}
= \frac{1}{4}\!\left[2^n + 2(\sqrt{2})^n \cos\!\left(\tfrac{n\pi}{4}\right)\right].
\]

Therefore,
\[
f(n) \;=\; \frac{1}{2} \;+\; 2^{-n/2}\cos\!\left(\tfrac{n\pi}{4}\right).
\]

Now, the only question remains whether the maximum value occurs at $n=7$ or not.   
\begin{itemize}
    \item For $n=1$,
    \[
    f(1) \;=\; \tfrac{1}{2} + 2^{-1/2}\cos\!\left(\tfrac{\pi}{4}\right) \;=\; 1,
    \]
    which is the global maximum.
    
    \item If we restrict to $n \geq 3$ (as is common in quantum secret sharing scenarios), then the maximum occurs when 
    $\cos\!\left(\tfrac{n\pi}{4}\right) = 1$, i.e., $n = 8k $, or when 
    $n = 7k, k \in \mathbb{N}$. But $2^{-n/2}$ is a decreasing function, hence we choose $k=1$.
    
    In particular, for $n=7$ or $n=8$,
    \[
    f(n) =  \frac{1}{2} + \frac{1}{16} = \frac{9}{16} = 0.5625,
    \]
    which is the maximum for $n \geq 3$.  \hfill 
\end{itemize}
\end{proof}
\section{Proof of Theorem 3}
\label{Proof of Theorem 3}
A natural yet significant question arises: why is the seven-qubit system decomposed into subsystems of dimensions $1$, $3$, and $3$? 
The conceptual reasoning behind this decomposition has been discussed as follows.
\begin{theorem}
    If a 7-qubit state is distributed among the three parties Alice, Bob, and Charlie in the partition \((1, j, 6-j)\), then the output views of the two participants, Bob and Charlie are indistinguishable when \(j = 3\).
\end{theorem}
\begin{proof}  The $7$-qubit GHZ state distributed among three parties 
(Alice with one qubit, Bob with $j$ qubits, and Charlie with $6-j$ qubits) is
\begin{equation*}
\ket{\text{GHZ}} 
= \frac{1}{\sqrt{2}}
\left(
\ket{0}_{A}\, \ket{0}^{\otimes j}_{B}\, \ket{0}^{\otimes (6-j)}_{C}
\;+\;
\ket{1}_{A}\, \ket{1}^{\otimes j}_{B}\, \ket{1}^{\otimes (6-j)}_{C}
\right).
\end{equation*}

The corresponding density operators are
\begin{align*}
\rho_{ABC} = \tfrac{1}{2}
&\Big[\left(
\ket{0}_{A}\, \ket{0}^{\otimes j}_{B}\, \ket{0}^{\otimes (6-j)}_{C}
\;+\;
\ket{1}_{A}\, \ket{1}^{\otimes j}_{B}\, \ket{1}^{\otimes (6-j)}_{C}
\right) \\
&\left(
\bra{0}_{A}\, \bra{0}^{\otimes j}_{B}\, \bra{0}^{\otimes (6-j)}_{C}
\;+\;
\bra{1}_{A}\, \bra{1}^{\otimes j}_{B}\, \bra{1}^{\otimes (6-j)}_{C}
\right)\Big].
\end{align*}

Tracing out $AB$ gives
\begin{equation*}
\rho_C = \mathrm{Tr}_{AB}(\rho_{ABC})
= \tfrac{1}{2}\left( \ket{0}\!\bra{0}^{\otimes j} + \ket{1}\!\bra{1}^{\otimes j} \right)_C .
\end{equation*}

Similarly, tracing out $AC$ gives
\begin{equation*}
\rho_B = \tfrac{1}{2}\left( \ket{0}\!\bra{0}^{\otimes (6-j)} + \ket{1}\!\bra{1}^{\otimes (6-j)} \right)_B. 
\end{equation*}   
Hence, $\rho_{B}$ and $\rho_{C}$ are identical only when $$j = 6-j \implies j=3.$$
\end{proof}
\begin{table*}[ht]
\begin{minipage}{\textwidth}
\centering
\renewcommand{\arraystretch}{1.2}
\setlength{\tabcolsep}{8pt}
\begin{tabular}{|c c c|c c c|c|}
\hline
\multicolumn{3}{|c|}{\textbf{Input}} & \multicolumn{3}{c|}{\textbf{Output}} & \textbf{Operation} \\
\hline
Alice & Bob & Charlie & Alice & Bob & Charlie & $K_A = K_B \oplus K_C$ \\
$x_1$ & $x_2, x_3, x_4$ & $x_5, x_6, x_7$ & $y_1$ & $y_2, y_3, y_4$ & $y_5, y_6, y_7$ &  \\
\hline

0 & 0,\;0,\;0 & 0,\;0,\;0 & 0 & 0,\;0,\;0 & 0,\;0,\;0 & $0 = 0 \oplus 0$ \\
 &   &   & 0 & 0,\;0,\;0 & 0,\;1,\;1 & $0 = 0 \oplus 0$ \\
  &   &   & 0 & 0,\;0,\;0 & 1,\;0,\;1 & $0 = 0 \oplus 0$ \\
 &   &   & $\vdots$ & $\vdots$ & $\vdots$ & $\vdots$ \\
 &   &   & 1 & 1,\;1,\;1 & 1,\;1,\;0 & $1 = 1 \oplus 0$  \\
\hline
 0 & 0,\;0,\;1 & 1,\;1,\;1 & 0 & 0,\;0,\;0 & 0,\;0,\;0 & $0 = 0 \oplus 0$ \\
 &   &   & 0 & 0,\;0,\;0 & 0,\;1,\;1 & $0 = 0 \oplus 0$ \\
  &   &   & 0 & 0,\;0,\;0 & 1,\;0,\;1 & $0 = 0 \oplus 0$ \\
 &   &   & $\vdots$ & $\vdots$ & $\vdots$ & $\vdots$ \\
 &   &   & 1 & 1,\;1,\;1 & 1,\;1,\;0 & $1 = 1 \oplus 0$  \\
\hline

  & $\vdots$ &  &  & $\vdots$ &  & $\vdots$   \\
  
\hline
1 & 1,\;1,\;1 & 0,\;0,\;0 & 0 & 0,\;0,\;0 & 0,\;0,\;0 & $0 = 0 \oplus 0$ \\
 &   &   & 0 & 0,\;0,\;0 & 0,\;1,\;1 & $0 = 0 \oplus 0$ \\
  &   &   & 0 & 0,\;0,\;0 & 1,\;0,\;1 & $0 = 0 \oplus 0$ \\
 &   &   & $\vdots$ & $\vdots$ & $\vdots$ & $\vdots$ \\
 &   &   & 1 & 1,\;1,\;1 & 1,\;1,\;0 & $1 = 1 \oplus 0$  \\
\hline
\end{tabular}
\caption{\label{H(A|B,C)}Updated design
 of DI-QSS scheme using 7-bit input. Our $(1,3,3)$ qubit configuration replicates with the Eq.\ref{hillery"s basis} \cite{PhysRevA.59.1829}, where Charlie can only predict that Alice and Bob's subsystems are correlated or not, but the individual results of Alice and Bob are completely hidden from his local outputs.}
\end{minipage}
\end{table*}


\section{Table 3: Updated design of DI-QSS scheme using $7$-bit input  } \label{Table 3: Updated design of DI-QSS scheme using $7$-bit input }

From Table~\ref{H(A|B,C)}, it is evident that for all inputs whose Hamming weight is of the form $4k$ $(k=0,1,2,\dots)$, the classical output bit on Alice’s side satisfies the XOR relation with the sum of the classical output bits of Bob and Charlie. In other words, Alice’s output $K_A$ is always equal to Bob's output plus Charlie's output \emph{i.e.} $K_B \oplus K_C$.


\section{$H(A|B,C) = 0$ for three parties} \label{$H(A|B,C) = 0$ for three parties}

{The output string $y=y_1\; y_2y_3y_4\; y_5y_6y_7$ has a special property (Table~\ref{H(A|B,C)}) that it contains an even number of 1's.
The output set remains invariant with respect to the inputs, implying that the input strings cannot be inferred from the outputs. For any $7$-bit input, the possible output set is identical, allowing any output string to occur with equal probability. }
 {\noindent Now consider two possible outputs:
\[ 
\underbrace{1}_{\text{Alice}}\quad
\underbrace{000}_{\text{Bob}}\quad
\underbrace{ 001}_{\text{Charlie}}
\quad \text{and} \quad
\underbrace{0}_{\text{Alice}}\quad
\underbrace{000}_{\text{Bob}}\quad
\underbrace{ 000}_{\text{Charlie}}
.\]

\begin{table*}[ht]
\renewcommand{\arraystretch}{1.1}
\centering
\begin{tabular}{|c|c|c|}
\hline
\textbf{Total Sum} $(\mathbf{\sum_i y_i})$ & \textbf{Partitions} $(\mathbf{y_1, \sum y_{_{B_i}}, \sum y_{_{C_i}}})$ &  $ \mathbf{y_1 = (\sum y_{_{B_i}}) \oplus (\sum y_{_{C_i}}) \pmod{2}}$ \\
\hline
\multirow{3}{*}{6} & $(1, 3, 2)$ & $1 = 3 \oplus 2 \equiv 1 \oplus 0 = 1$  \\
 & $(1, 2, 3)$ & $1 = 2 \oplus 3 \equiv 0 \oplus 1 = 1$  \\
 & $(0, 3, 3)$ & $0 = 3 \oplus 3 \equiv 1 \oplus 1 = 0$  \\
\hline
\multirow{7}{*}{4} & $(1, 3, 0)$ & $1 = 3 \oplus 0 \equiv 1 \oplus 0 = 1$  \\
 & $(1, 0, 3)$ & $1 = 0 \oplus 3 \equiv 0 \oplus 1 = 1$  \\
 & $(0, 2, 2)$ & $0 = 2 \oplus 2 \equiv 0 \oplus 0 = 0$  \\
 & $(0, 1, 3)$ & $0 = 1 \oplus 3 \equiv 1 \oplus 1 = 0$  \\
 & $(0, 3, 1)$ & $0 = 3 \oplus 1 \equiv 1 \oplus 1 = 0$  \\
 & $(1, 1, 2)$ & $1 = 1 \oplus 2 \equiv 1 \oplus 0 = 1$  \\
 & $(1, 2, 1)$ & $1 = 2 \oplus 1 \equiv 0 \oplus 1 = 1$  \\
\hline
\multirow{5}{*}{2} & $(1, 0, 1)$ & $1 = 0 \oplus 1 \equiv 0 \oplus 1 = 1$  \\
 & $(1, 1, 0)$ & $1 = 1 \oplus 0 \equiv 1 \oplus 0 = 1$  \\
 & $(0, 1, 1)$ & $0 = 1 \oplus 1 \equiv 1 \oplus 1 = 0$  \\
 & $(0, 2, 0)$ & $0 = 2 \oplus 0 \equiv 0 \oplus 0 = 0$  \\
 & $(0, 0, 2)$ & $0 = 0 \oplus 2 \equiv 0 \oplus 0 = 0$  \\
\hline
0 & $(0, 0, 0)$ & $0 = 0 \oplus 0 \equiv 0 \oplus 0 = 0$  \\
\hline
\end{tabular}
\caption{Partitioning of Output Sums $(\sum_i y_i)$ into Alice, Bob, and Charlie's.}
\label{partition}
\end{table*}
 Both the outputs have an even number of $1$'s. 
Here, the first entry corresponds to the output of Alice’s device, the second block of three bits corresponds to the output of Bob’s device, and the third block of three bits corresponds to the output of Charlie’s device.}
If we compare these two outputs, it becomes evident that Bob’s portion of the output remains identical in both cases. This means that Bob’s information alone is insufficient to distinguish between the two possible global outputs. 
In other words, Bob cannot, by himself, predict or determine the complete outcome of the system. The overall output depends on the combined contributions of all three parties, and no single participant has enough information to reconstruct the global output independently. \\
  Now from Charlie's point of view suppose we consider two different outputs 
  \[ 
\underbrace{1}_{\text{Alice}}\quad
\underbrace{001}_{\text{Bob}}\quad
\underbrace{ 000}_{\text{Charlie}}
\quad \text{and} \quad
\underbrace{0}_{\text{Alice}}\quad
\underbrace{101}_{\text{Bob}}\quad
\underbrace{ 000}_{\text{Charlie}}
\]
{Both the outputs contain an even number of $1$’s, and hence they are valid outputs. From Charlie’s perspective, however, his share of the outputs is the same bitstring $000$ in both cases. This information alone is insufficient for him to reconstruct or infer the complete output string}. This scenario is identical to that observed in Hillery’s bases (Fig.~\eqref{Hillery's basis}), where neither Charlie’s nor Bob’s basis alone can determine Alice’s outcome; they can only guess whether the outcomes are correlated or anticorrelated.

For \( y = y_1\, y_2y_3y_4\, y_5y_6y_7 \), the condition for validity requires that the sum of the output bits be even, i.e., \( \sum_i y_i \) is even. The possible values of this sum are \( \sum_i y_i \in \{0, 2, 4, 6\} \). The corresponding partitions among Alice, Bob, and Charlie’s subsystems can be expressed as  $(y_1, \sum y_{_{B_i}} = y_2 + y_3 + y_4, \sum y_{_{C_i}} = y_5 + y_6 + y_7)$, and they are shown in Table~\ref{partition}.
In this case, the XOR of Charlie’s output yields $0$ or $1$, guessing whether it is correlated or anticorrelated with Alice’s outcome. Hence, for Bob, it follows that $H(A|B) = \tfrac{1}{2}$, an analogous relation holds for Charlie as well. But whenever Bob and Charlie meet and combine their outcomes, because for any partition $(a,b,c)$ shown above $a = b \oplus c \pmod{2}$ is satisfied. Consequently, $H(A|B,C) = 0$ is the required result, and it follows for all output bits generated from the proposed multiparty pseudo-telepathy based DI-QSS scheme.

\section{Proof of Theorem 4}
\label{Proof of Theorem 4}
\begin{theorem}
   The proposed device-independent secret-sharing scheme is $\varepsilon_{\mathrm{cor}}$-correct, where
$\varepsilon_{\mathrm{cor}} > 1 - X$, and
\[
X = \binom{R_{l}}{(1-\nu)R_{l}}\, p_{m}^{(1-\nu)R_{l}} (1-p_{m})^{\nu R_{l}} ,
\]
where
\begin{itemize}
    \item $R_{l}$ denotes the number of rounds fixed by the dealer to obtain a final key of length~$l$,
  \item $\nu$ denotes the allowed proportion of deviation, and 
    \item $p_{m}$ is the probability that Alice’s bit matches Bob${ }\oplus{ }$Charlie’s output after one round of the protocol.
 \end{itemize}
\end{theorem}

\begin{proof} Let $\nu$ be the proportion of deviation allowed.   

The correctness definition states that the dealer and the two participants must satisfy 
\begin{equation}y_A = K_B \oplus K_C \label{key_gen_cond}\end{equation}
If $R_l$ is the number of rounds Alice wants to run the protocol, where $l$ is the desired key length, then in those $R-l$ rounds, we need to identify where the Eq.~\eqref{key_gen_cond} satisfies and where there is a mismatch.
So, $(1-\nu)R_l$ is the total cases where Eq.~\eqref{key_gen_cond} satisfies and $\nu R_l$ is the cases where condition $(2)$ not holds.

$\therefore$ The probability of matching with Eq.~\eqref{key_gen_cond} after $R_l$ rounds is  \[ \binom{R_{l}}{(1-\nu)R_{l}}\, p_{m}^{(1-\nu)R_{l}} (1-p_{m})^{\nu R_{l}} \]
Now, from the definition of Correctness 
 \begin{align*} \binom{R_{l}}{(1-\nu)R_{l}}\, p_{m}^{(1-\nu)R_{l}} &(1-p_{m})^{\nu R_{l}} > 1 - \varepsilon_{\mathrm{cor}}\\
 {\implies \quad } \varepsilon_{\mathrm{cor}} > 1 - X  \; ,
 \end{align*}
 where $X = \binom{R_{l}}{(1-\nu)R_{l}}\, p_{m}^{(1-\nu)R_{l}} (1-p_{m})^{\nu R_{l}}.$

\end{proof}

\section{QBER under Photon Loss in a Noisy scenario}
\label{QBER calculation}
We now evaluate the effect of photon loss on the DI-QSS protocol, noting that the valid outputs always contain an even number of $1$'s. First, we consider the case of no photon loss.\\
$(i)$ \noindent\text{No photon loss:}
In this case, the probability that the $7$ qubit output string contains an even number of \(1\)s is
\begin{align*}
\Pr[\text{even number of }1\text{s}]
    &= \left(\tfrac{1}{2}\eta\right)\left(\tfrac{1}{2}\eta\right)\cdots\left(\tfrac{1}{2}\eta\right)\,\eta  \notag \\
    &= \left(\tfrac{1}{2}\eta\right)^{6}\eta .
\end{align*}
Now the number of such output bits are $\frac{2^7}{2} = 2^6$. Therefore, the overall probability of obtaining an even-parity output in the no photon loss scenario is $2^6 \cdot\left(\frac{\eta^7}{2^6} \right) =\eta^7$.

The quantum bit error rate (QBER) is defined as
\begin{equation*}
\text{QBER} = \frac{\text{Number of erroneous bits}}{\text{Total number of transmitted bits}} .
\end{equation*}

If we restrict attention to the events in which photon loss occurs, then there are $7$ distinct loss configurations, including the case in which all photons are lost. So,\\
$(ii)$ \noindent\text{One photon loss:} Here the probability is 
\begin{align}
\Pr[\text{one photon lost}]
    &= \left(\tfrac{1}{2}\eta\right)\left(\tfrac{1}{2}\eta\right)\cdots\left(\tfrac{1}{2}\eta\right)\overline{\eta}  \notag \\
    &= \left(\tfrac{1}{2}\eta\right)^{6}\bar{\eta} .
\end{align}
The number of such output bits are $7\cdot2^6$.\\
$(iii)$ \noindent\text{Two photon loss:} The probability is 
\begin{align}
\Pr[\text{two photon lost}]
    &= \left(\tfrac{1}{2}\eta\right)^5\,\bar{\eta}. 
\end{align}
The number of such output bits are $\binom{7}{2}\cdot2^5$.

Proceeding in this manner, we obtain a total of $7$ cases, including the scenario where all photons are lost. 
Hence, combining all these, where photon lost happens the QBER $Q_2$ is $$Q_2 =  \frac{1}{2}[(\eta + \bar{\eta})^7 - \eta^7] = \frac{1}{2}(1- \eta^7) .$$
where the $\frac{1}{2}$ implies the error probability. Since the key relation is a single binary constraint, a uniformly random bit(after photon loss) violates that with probability $\frac12$.

\section{Lower Bound of $H(A|E)$ under Noisy Scenario} \label{Lower Bound of $H(A|E)$ under Noisy Scenario}
Alice, Bob, and Charlie generate raw key bits at their respective locations, with
\[
K_A \in \{0,1\}, \quad K_B, K_C \in \{0,1\}.
\]
Eve holds a quantum system \(E\) that may be arbitrarily entangled with the joint system $(A,B,C)$.

The key-generation condition of the protocol is defined by the parity relation
\[
K_A = K_B \oplus K_C.
\]
An error occurs whenever this relation is violated, i.e., when $K_A \neq K_B \oplus K_C$.

For binary random variables, the conditional Shannon entropy of Alice’s bit given Bob’s and Charlie’s bits is completely characterized by the error rate $Q$. In particular,
\[
H(A|B,C) = h(Q),
\]
where $h(\cdot)$ denotes the binary entropy function.

Next, we invoke the monogamy of quantum correlations~\cite{PhysRevA.73.012112, PhysRevA.61.052306}, which can be expressed as an entropic uncertainty relation with quantum side information. For the tripartite system $(A,B,C)$ and Eve’s quantum system $E$, this yields
\begin{equation}
H(A |E) + H(A|B,C) \geq H(A).
\end{equation}

Since Alice’s raw key bit is uniformly random by construction of the protocol, we have $H(A) = 1$.

Substituting $H(A |B,C) = h(Q)$ into the above inequality, we obtain
\[
H(A| E) + h(Q) \geq 1,
\]
which finally leads to the bound
\begin{equation*}
H(A|E) \geq 1 - h(Q).
\end{equation*}

\section{Mixture of the GHZ states} \label{Mixture of the GHZ states}
The mixed state of Alice, Bob and Charlie for a $7$-qubit state, including all noise effects is 
\begin{widetext}
\begin{align*}
\rho_{ABC}
&= \eta^{7}\!\Big(
F\,|GHZ^{+}\rangle\langle GHZ^{+}|
+\frac{1-F}{2^{7}}\, \mathbb{I}_{2^{7}}
\Big) \nonumber\\
&\quad + \sum_{k=1}^6 \binom{7}{k}\eta^{(7-k)}(\bar{\eta})^k \frac{1}{2}\Big(
| 0^{\otimes (7-k)}\rangle\langle 0^{\otimes (7-k)}| 
+| 1^{\otimes (7-k)}\rangle\langle  1^{\otimes (7-k)}|
\Big) \nonumber \\
&\quad + \bar{\eta}^{7}\,
|{\rm vac}\rangle\langle{\rm vac}|. 
\end{align*}
\end{widetext}
$\mathbb{I}_{2^{7}}$ represents the identity operator corresponding to the
uniform mixture over all $2^{7}$ seven-qubit GHZ basis states. The entire expression of $\mathbb{I}_{2^{7}}$  is given below
\begin{equation}
    \mathbb{I}_{2^7} = \sum_{j=1}^{64} \left( \ket{\text{GHZ}_j^+} \bra{\text{GHZ}_j^+} + \ket{\text{GHZ}_j^-} \bra{\text{GHZ}_j^-} \right), \tag{A2}
\end{equation}
where

\begin{align*} \ket{\text{GHZ}_1^{\pm}} &= \frac{1}{\sqrt{2}} (\ket{0000000} \pm \ket{1111111}) \\
\ket{\text{GHZ}_2^{\pm}} &= \frac{1}{\sqrt{2}} (\ket{0000001} \pm \ket{1111110}) \\
&\vdots \\
\ket{\text{GHZ}_{64}^{\pm}} &= \frac{1}{\sqrt{2}} (\ket{0111111} \pm \ket{1000000}) .
\end{align*}
In the white-noise model, the target GHZ state ($\ket{\text{GHZ}_1^{+}}$) will
degrade to one of the states given above with equal probability $\frac{1-F}{2^7}$.


\end{document}